\documentclass[conference]{IEEEtran}
\title{$O(\log \log n)$ Worst-Case Local Decoding and Update  Efficiency for Data Compression}
\author{
	\IEEEauthorblockN{
		Shashank Vatedka\IEEEauthorrefmark{1}, 
		Venkat Chandar\IEEEauthorrefmark{2}, 
		Aslan Tchamkerten\IEEEauthorrefmark{3}
	}
	\IEEEauthorblockA{
		\IEEEauthorrefmark{1}Dept.\ of Electrical Engineering, Indian Institute of Technology Hyderabad, India\\
		\IEEEauthorrefmark{2}DE Shaw, New York, USA\\
		\IEEEauthorrefmark{3}Dept.\ of Communications and Electronics, Telecom Paris, France 
	}
}
\usepackage{basicreq}

\newcommand{\poly}{\mathrm{poly}}

\begin{document}
\maketitle

\begin{abstract}
  This paper addresses the problem of data compression with local decoding and local update. A compression scheme has worst-case local decoding $ \rwc $ if any bit of the raw file can be recovered by probing at most $ \rwc $ bits of the compressed sequence, and has update efficiency of $\twc  $ if a single bit of the raw file can be updated by modifying at most $ \twc $ bits of the compressed sequence. This article provides an entropy-achieving compression scheme for memoryless sources that simultaneously achieves $ O(\log\log n) $ local decoding and update efficiency. Key to this achievability result is a novel succinct data structure for sparse sequences which allows efficient local decoding and local update.
  
  Under general assumptions on the local decoder and update algorithms, a converse result shows that $ \rwc $ and $ \twc $ must grow as $ \Omega(\log\log n) $.
\end{abstract}

\section{Introduction}
Consider a source sequence $ X^n $ with independent and identically distributed (i.i.d.) components having probability mass function $ p_X $ on a finite alphabet $ \cX $. For simplicity we assume here that $ p_X $ is a known distribution over $ \cX=\{0,1\} $.\footnote{This can be generalized to the scenario where $ p_X $ is unknown to the encoder and decoder by first estimating $ p_X $ and then using this for designing the compression scheme as in \cite{vatedka2019local}. Similarly, our results can be generalized to nonbinary alphabets.}

 A fixed-length compression scheme of rate $ R $ consists of a pair of algorithms, the encoder $\enc$ and the decoder $\dec$. The encoder maintains for every $X^n$, a codeword $C^{nR}\in\{0,1\}^{nR}$ such that $\dec(C^{nR})$ is a good estimate of $X^n$. 
 The probability of error of the compression scheme is defined as 
\[
\peglob\coloneq \Pr[\dec(C^{nR})\neq X^n].
\]
From the source coding theorem, we know that there exist sequences of codes with rate arbitrarily close to the entropy $ H(p_X) $ and error probability vanishing in $ n $.

Our goal is to design a fixed-length compression scheme that additionally supports local encoding and decoding. A locally decodable and updatable compression scheme consists of a global encoder and decoder pair $ (\enc,\dec) $ and in addition, a local decoder and a local updater:
\begin{itemize}
	\item {\bf Local decoder:} A local decoder is an algorithm which given $ i\in[n] $, probes (possibly adaptively) a small number of bits of $ C^{nR} $ to output $ \widehat{X}_i $. Here, $ \widehat{X}_i $ is the $ i $th symbol of $\widehat{X}^n\defeq \dec(C^{nR}) $. The worst-case local decodability $ \rwc $ of the scheme is the maximum number of bits probed by the local decoder for any $ i $. The average local decodability $ \ravg $ is the expected number of bits (averaged over the source distribution) probed to recover any $ \widehat{X}_i $.
	\item  {\bf Local updater:} A local updater is an algorithm which given $ i\in[n] $ and $ \widetilde{X}_i\in\{0,1\} $, adaptively reads and modifies a small number of bits of $ C^{nR} $ to give $ \widetilde{C}^{nR} $ such that $ \widetilde{C}^{nR}=
	\enc(X_1,\ldots,X_{i-1}\widetilde{X}_i,X_{i+1},\ldots,X_n) $. In particular, the local updater has no prior knowledge of $X^n$ or $C^{nR}$ and must probe $C^{nR}$ to obtain such information. 
	
	The worst-case update efficiency $ \twc $ is defined as the maximum of the sum of the bits read and written in order to update any $ i $. Also, we assume that $ \widetilde{X}_i $ is distributed according to $ p_X $ and is independent of $ X^n $. Likewise, the average update efficiency $ \tavg $ is the sum of the average number of bits probed and written in order to update any~$ X_i $.
\end{itemize}
It was recently shown in~\cite{vatedka2019local} that $$ (\ravg,\tavg)=(O(1),O(1)) $$ is achievable. In that paper the authors also gave a separate compression scheme  achieving\footnote{Throughout the paper, we use logarithms to base two.} $$ (\rwc,\tavg)=(O(\log\log n),O(\log\log n)).$$ In particular, the question of whether $$ (\rwc,\twc)=(O(\log\log n),O(\log\log n)) $$ is achievable was left open. In this paper we answer this question in the affirmative.  We also show that under certain additional assumptions on the local decoder and the local updater this locality is order optimal.

Our achievability proof is based on a novel succinct data structure for  $O(b/\log b)$-sparse sequences of length $b$ in the bitprobe model which for any $0<\delta<1$ takes space $O(\delta b)$ while enabling local decode and update using at most $O(\log b)$ and $O(\frac{1}{\delta}\log b)$ bit reads/writes respectively.
Our restricted converse is based on an analysis of bipartite graphs that represent the encoding and decoding algorithms.

\subsection{Prior work}

Local decoding and update for entropy-achieving compression schemes have been studied mostly in isolation.
The problem of locally decodable source coding of random sequences has received attention very recently following~\cite{makhdoumi2013locally-arxiv,makhdoumi_onlocallydecsource}.
Mazumdar \emph{et al.}~\cite{mazumdar2015local} gave a fixed-length compressor of rate of $ H(p_X)+\varepsilon $ with $\rwc(1)= \Theta(\frac{1}{\varepsilon}\log\frac{1}{\varepsilon}) $. They also provided a converse result for non-dyadic sources: $ \rwc(1)=\Omega(\log(1/\varepsilon)) $ for any compression scheme that achieves rate $ H(p_X)+\varepsilon $. 
Similar results are known for variable length compression~\cite{pananjady2018effect} and universal compression of sources with memory~\cite{tatwawadi18isit_universalRA}. Likewise, there are compressors  that achieve~\cite{montanari2008smooth} rate $R=H(p_X)+\varepsilon$ and update efficiency $\twc=O(1)$.

In the computer science community, the literature has mostly focused on the word-RAM model~\cite{patrascu2008succincter,patrascu2014dynamic,makinen2006dynamic,sadakane2006squeezing,navarro2014optimal,viola2019howtostore}, where each operation (read/write/arithmetic operations) is on words of size $ w=O(\log n) $ bits each, and the complexity is measured in terms of the number of word operations required for local decoding/update. 
However, in this case the number of bitprobes required is $ \Omega(\log n) $. For random messages, it is almost trivial to obtain local decoding/update efficiency of $ O(\log n) $ bitprobes by partitioning the $ n $ message symbols into blocks of size $ O(\log n) $ and compressing each block separately. 

\section{Contributions}
Before we present our main results we make a few observations aimed at justifying our model, and in particular the requirements we impose on the local decoder (see previous section).

Note that for any index $i$ the local decoder must output $ \widehat{X}_i $, the $i$-th estimate of the global decoder. We could potentially relax this constraint, by requiring that the local decoder produces an estimate $ \widehat{X}^{(\mathrm{loc})}_i $, potentially different from $ \widehat{X}_i $,  with small error probability $$ \peloc\defeq \max_i\Pr[\widehat{X}^{(\mathrm{loc})}_i\neq X_i]. $$ 
As we argue next, if we impose $\peloc$ to be small (but non-vanishing) with either no global decoder or with a separate global decoder that achieves vanishing error probability, then from a coding perspective the solution is essentially trivial.
\subsubsection{Only local decodability and non-vanishing error}
If we only require $\peloc $ to be small without constraints on the  global decoder, then we can easily achieve $$ (\rwc,\twc)=\left(O\left(\log\frac{1}{\peloc}\right),O\left(\log\frac{1}{\peloc}\right)\right) .$$ This can be obtained by partitioning the length$-n$ message into blocks of size $b_0= O(\log\frac{1}{\peloc}) $, and compressing each block using an entropy-achieving fixed-length compression scheme---notice that the probability of wrongly decoding any particular block vanishes exponentially with~$ b_0 $. Hence, for any small but constant $\peloc$ we can achieve $$(\rwc,\twc)=(O(1),O(1)).$$
\subsubsection{Separate local and global decoders}
Suppose that in addition to $1)$ we also want $ \peglob=o(1) $ using a separate global decoder to recover $ \widehat{X}^n $. This can be obtained by using a low-density parity check (LDPC) code with $O(1)$ maximum variable and check node degrees. The codeword consists of two parts: $$C^{nR}=(C^{n(R-\delta)}(1),C^{\delta n}(2)),$$ where $C^{n(R-\delta)}(1)$ is obtained as in the previous case by dividing the message into constant size $b_0$ blocks and separately encoding each, while $C^{\delta n}(2)$  is obtained as the syndrome (of the LDPC code) of the (Hamming) error vector between $X^n$ and the decoding of $C^{n(R-\delta)}(1)$. 
	
	The local decoder only probes $C^{n(R-\delta)}(1)$, while the local updater needs to update both $C^{n(R-\delta)}(1)$ and $C^{n\delta}(2)$. Since we are using an LDPC code, $C^{n\delta}(2)$ can be updated  using $O(1)$ bit modifications.  Therefore, $$(\rwc,\twc)=(O(\log\frac{1}{\peloc}),O(\log\frac{1}{\peloc})).$$ 
	
	The global decoder decodes both $C^{n(R-\delta)}(1)$ and $C^{\delta n}(2)$ and can recover $X^n$ with $o(1)$ probability of error.
	


	As we see, $1)$ and $2)$ are essentially trivial cases from a coding perspective. Note also that the solution to $2)$ generally depends on the rate of decay of $\peloc$.
Requiring the local decoder to output $ \widehat{X}_i $ removes this degree of freedom---since $\peglob=o(1)$ implies $\peloc=o(1)$---and, as we argue below, is more interesting from a coding perspective. This setup is perhaps more interesting also practically since we can parallelize  global decoding. If we have a large number of parallel processors, then the runtime of global decoding can be made sublinear in $ n $. 
	
	We will henceforth only consider compression schemes with local decoders that output $ \widehat{X}_i $.

The main result of this article is the following:
\begin{theorem}\label{thm:ach_ologlogn}
	For any $\varepsilon>0$ there exists a compression scheme for Bernoulli($ p $) sources that achieves 
	\[
	(R,\rwc,\twc)= \left( H(p)+\epsilon,O(\log\log n),O(\frac{1}{\epsilon}\log\log n) \right),
	\] 
	and the overall computational complexity of global encoding/decoding is quasilinear in $ n $.
\end{theorem}
The above theorem is formally proved in Section~\ref{sec:prf_ologlogn_ach}. The proof of the above theorem is based on a novel dynamic succinct data structure for sparse sequences that achieves $ O(\log n) $ locality in the bitprobe model. 
\begin{lemma}\label{lemma:succ_ds_ologn}
	Fix any $ \delta>0 $. For every $ \beta = o(b/\log b) $, there exists a dynamic succinct data structure for $ b $-length binary vectors of sparsity at most $ \beta $ with the following properties. Any such vector occupies at most $ \delta b(1+o(1)) $ bits, has worst-case local decoding $ O(\log b) $ and worst-case update efficiency at most $ O(\frac{1}{\delta}\log b) $. 
\end{lemma}

To prove a lower bound, we make three assumptions 
\begin{itemize}
	\item[(A1)] \emph{Global encoding and decoding using local algorithms:} We assume that  $ C^{nR} $ is obtained by running the local  updater on each message symbol, and $\widehat{X}^n$ by running the local decoder for each bit. In other words, there is no separate global encoder or decoder. 
	\item[(A2)] \emph{Function assumption:} 
	The local update function for the $ t $'th update is a deterministic function of $ X^n $, and is independent of the sequence of the previous $ t-1 $ updates. 
	Likewise, the local decoder does not depend on the sequence of updates that have occurred previously. 
	\item[(A3)] \emph{Bounded average-to-worst case influence:} For nonadaptive schemes, we can construct the corresponding local encoding graph $\cG_e$ and the local decoding graph $\cG_d$ as follows---see Fig.~\ref{fig:encoding_decoding_graph}. Two vertices $(i,j)$ in $\cG_e$ are adjacent if $C_j$ is a function of $X_i$. Likewise, $(j,i)$ are adjacent in $\cG_d$ if $\widehat{X}_i$ is a function of $C_j$. We assume that the ratio of the average to the worst case degrees of the right vertices of $\cG_e$ is bounded from below by a constant independent of $n$, and similarly for the left vertices of $\cG_d$. 
\end{itemize}
\begin{figure}
	\begin{center}
		\includegraphics[width=7cm]{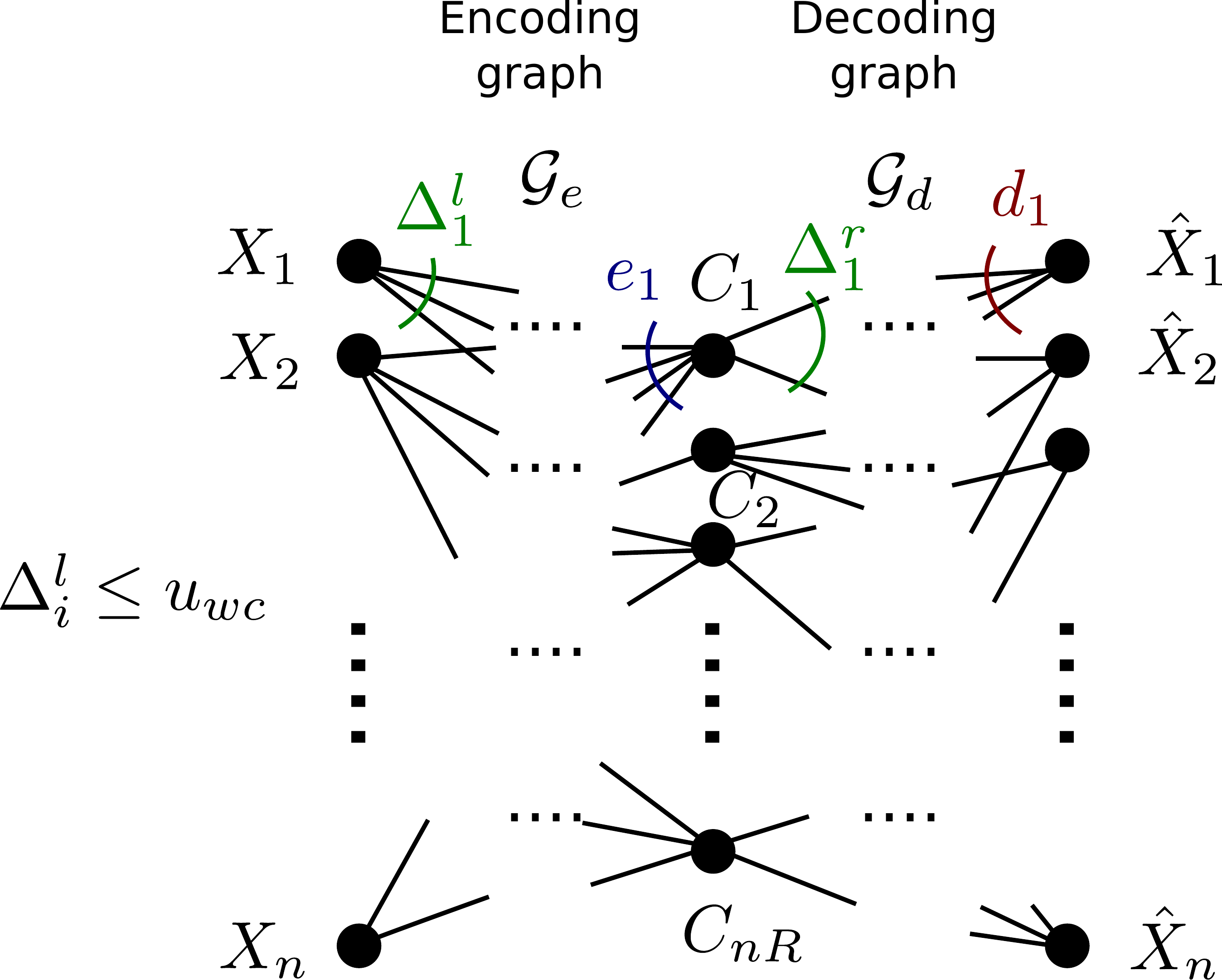}
		\caption{
		Encoding and decoding graphs under assumptions (A1) and (A2). The degree of the $ i $th left vertex in $ \cG_e $ is $ \Delta^l_i $, which is upper bounded by $u_i$, the update efficiency for message symbol $i$. The $ j $th right vertex is the local encodability of the $ i $th symbol $ e_j $, defined in Sec.~\ref{sec:lbound_simult_loc}. The degree of a right vertex in $ \cG_d $ is equal to the local decoding of the $ i $th symbol $ d_i $.}
		\label{fig:encoding_decoding_graph}
	\end{center}
	\vspace{-0.5cm}
\end{figure}
\noindent {\it{Remark:}} Under (A1) and (A2), any adaptive scheme (whether for update or decoding) achieving worst-case locality $ l $ can be converted to a non-adaptive scheme with locality $\leq 2^l $. This is because an adaptive scheme with locality $ l $ may depend on up to $ 2^l $ different bits. Therefore a lower bound on the locality of a non-adaptive scheme translates into a lower bound on the logarithm of the locality of an adaptive scheme.

\begin{theorem}\label{thm:dwc_uwc_lbound_adaptive}
	For any adaptive scheme that satisfy (A$1$)--(A$3$) with $R<1$, we have
	\[
	\rwc + \twc = \Omega(\log\log n).
	\]
\end{theorem}


As mentioned earlier, it is possible to achieve $ (\rwc,\twc)=(O(1),O(1)) $ without assumptions (A1)--(A3). We conjecture that Theorem~\ref{thm:dwc_uwc_lbound_adaptive} holds even without assumption (A3), but were unsuccessful in proving this. We also conjecture that it holds even without (A2). As we will see later, the data structure that leads to Lemma~\ref{lemma:succ_ds_ologn} does not satisfy (A2).

\section{Achievability}

The high-level structure of our scheme  is inspired by the locally decodable compressor in~\cite{mazumdar2015local}. The idea in~\cite{mazumdar2015local} is to partition the set of message symbols into constant-sized blocks and use a fixed-length compressor for each block. The residual error vector is then encoded using the succinct data structure in~\cite{buhrman2002bitvectors} that occupies negligible space but allows local decoding of a single bit using $O(1)$ bitprobes. However, this data structure is static, in the sense that it does not allow efficient updates and hence does not get us small $\twc$.

A well-known dynamic data structure in the word-RAM model is the van Emde Boas tree~\cite{van1975preserving} which takes space $O(b)$ but allows local retrieval, insert and delete in $O(\log\log b)$ time (equivalently $O(\log b\,\log\log b)$ bitprobes). If we use the van Emde Boas tree for encoding the residual error vector, then we can achieve rate close to $H(p_X)$ but a higher locality of $O(\log\log n\;\poly(\log\log\log n))$.

Our main contribution is a novel dynamic succinct data structure (used to encode the residual error) which occupies roughly the same space as~\cite{buhrman2002bitvectors} but allows local decode \emph{and} update using only $O(\log b)$ bitprobes.

\subsection{Proof of Theorem~\ref{thm:ach_ologlogn}}\label{sec:prf_ologlogn_ach}

We partition the $ n $-length message sequence $ x^n $ into blocks $ x^{b_1}(1),\ldots,x^{b_1}(n/b_1) $ of size $ b_1 =O(\log n) $ each. Each block $ i $ is further partitioned into subblocks ($ x^{b_0}(i,1),\ldots, x^{b_0}(i,b_1/b_0) $) of $ b_0 $ symbols each.
Each subblock is compressed independently using a fixed-length lossy compression scheme of rate $ H(p)+\epsilon $ and average per-letter distortion $ \epsilon $. Let $ c^{b_0(H(p)+\epsilon)}(i,j) $ denote this subcodeword for the $ (i,j) $th subblock. In addition, the error vectors 
(denoted $ e^{b_0}(i,j) $ and equal to $ x^{b_0}(i,j) $ if $ x^{b_0}(i,j) $ is atypical and $ 0^{b_0} $ otherwise) 
are concatenated and for each $ i $, $e^{b_1}(i)\defeq (e^{b_0}(i,1),\ldots, e^{b_0}(i,b_1/b_0)) $ is compressed using the scheme in Lemma~\ref{lemma:succ_ds_ologn} to give $ \bar{c}^{\epsilon{b_1}}(i) $ with $ \delta=\epsilon $. 
The overall codeword is the concatenation of $(c^{b_0(H(p)+\epsilon)}(i,j): 1\leq i\leq n/b_1,1\leq j\leq b_1/b_0)$ and $(\bar{c}^{\epsilon{b_1}}(i):1\leq i\leq n/b_1)$.

As long as the sparsity of $ e^{b_1}(i) $ is less than  $ \alpha b_1/\log b_1 $ for a suitably chosen $\alpha>0$, we can recover the $i$th message block (or any symbol within it) without error.

We choose $ b_1 = \alpha_1(\log n\log\log n) $ and $ b_0=\alpha_2(\log\log n) $.

Using Azuma's inequality (and carefully choosing $\alpha_1,\alpha_2$), the probability that the distortion in each block is greater than $ \alpha\log n/\log\log n $ falls as $ o(1/n) $.

The worst-case local decoding is at most $ O(\log b_1)=O(\log\log n) $, while the worst-case update efficiency is $ O(\frac{1}{\epsilon}\log\log n) $. The compression rate is $ H(p)+2\epsilon $, and the overall probability of error (using the union bound over blocks) is $ o(1) $. This completes the proof. \qed

All that remains is to prove Lemma~\ref{lemma:succ_ds_ologn}.


\subsection{A succinct data structure achieving $ O(\log b) $ locality for sparse sequences of length $ b $}\label{sec:scheme_succdata_ologb}

\begin{figure*}
    \centering
    \includegraphics[width=13cm]{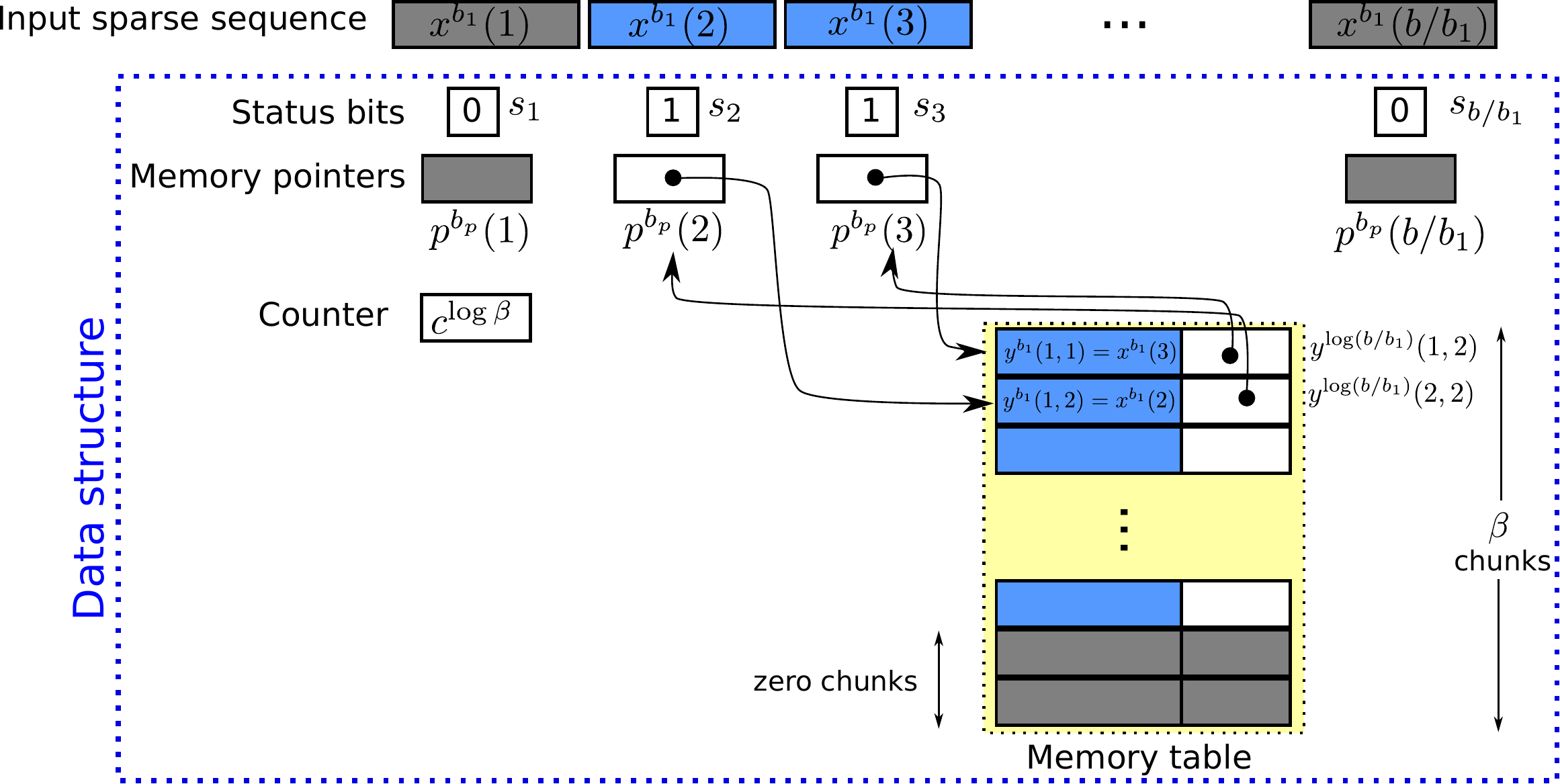}
    \caption{Illustrating the data structure in Sec.~\ref{sec:scheme_succdata_ologb}. For each message block, we store a status bit to denote whether the block is nonzero, and a memory pointer indicating the corresponding location in the memory table. Additionally, we store a counter with the number of nonzero blocks and a memory table. Each chunk in the memory table contains some data and a reverse pointer to the corresponding block. In the figure, nonzero blocks are colored blue while zero blocks/chunks are colored grey. To decode a nonzero block (say the $3$rd block in this example), the algorithm reads the memory pointer, and outputs the corresponding chunk that this points to (the blue chunk in the first row of the memory table). It is to be noted that after multiple updates, the order in which blocks are stored in the memory table could be arbitrary. For \emph{e.g.}, the first chunk in the memory table need not store the first nonzero block as illustrated in the figure above.}
    \label{fig:new_datastructure}
\end{figure*}

The high-level idea in our data structure is to split the $b$ symbols into blocks of $O(\log b)$ symbols each, and maintain a dynamic memory table where we only store the blocks with nonzero Hamming weight. Addressing is resolved by storing for each block a pointer which indicates the location in the memory table where the block is encoded.

Suppose that we split the $ b $-length sequence $ x^b $ into blocks of $ b_1 $ consecutive symbols each: $ x^{b_1}(1),\ldots, x^{b_1}(b/b_1) $.

The data structure has the following components:
\begin{itemize}
	\item \emph{Status bits:} These indicate whether each block is zero or not. We store $ b/b_1 $ many bits $ s_1,\ldots,s_{b/b_1} $, one for each block.  The status bit $ s_i $ is set to $ 1 $ if the Hamming weight of $ x^{b_1}(i) $ is greater than zero, and set to zero otherwise.
	\item \emph{Memory table:} The table stores information about the nonzero blocks. The table is organised into $ \beta $ many chunks\footnote{The chunks refer to the virtual partition of the memory table into units of $b_m$ bits each. This is to distinguish this from the ``blocks'' which form partitions of the $b$-length input sequence.} $ y^{b_m}(1),\ldots, y^{b_m}(\beta) $ of $b_m= b_1+\log(b/b_1) $ bits each.  The $ i $th chunk $ y^{b_m}(i) $ is further split into two parts: $ y^{b_1}(i,1) $ having $ b_1 $ bits, and $ y^{\log(b/b_1)}(i,2) $ having $ \log(b/b_1) $ bits. Here, $ y^{b_1}(i,1) $ is a vector which stores a nonzero block $ x^{b_1}(j) $ for some (suitably defined later) $ j $, while $ y^{\log(b/b_1)}(i,2) $ is a reverse pointer which encodes $ j $ in $ \log(b/b_1) $ bits. 
	\item \emph{Memory pointers:} These indicate  where each block is stored in the memory table. There are $ b/b_1 $ pointers $ p^{b_{p}}(1),\ldots, p^{b_{p}}(b/b_1) $ of $ b_{p}=\log \beta $ bits each, one for each block. 
	\item \emph{Counter for number of nonzero blocks:} $c^{\log \beta}$ is a vector of length $ \log \beta $ bits which stores the number of nonzero blocks in $ x^b $.
\end{itemize}
The overall codeword is a bit sequence obtained by the concatenation of the status bits, memory table, memory pointers and the counter. The total space required  is
\begin{align}
k_{ds} &= \frac{b}{b_1}+\beta(b_1+\log(b/b_1)) + \frac{b}{b_1}\log\beta + \log\beta \label{eq:space_newds}
\end{align}
The data structure is illustrated in Fig.~\ref{fig:new_datastructure}.

\subsubsection{Initial encoding}
Let $ k $ denote the number of nonzero blocks in $ x^b $.
\begin{itemize}
	\item Status bits: If $ x^{b_1}(i) $ has nonzero Hamming weight, then $ s_i=1 $. Otherwise, it is set to zero.
	\item Memory pointers: 
	If $x^{b_1}(i)=0^{b_1}$, then  $p^{b_p}(i)=0^{b_p}$.
	If not, and $x^{b_1}(i)$ is the $ j $th nonzero block among $ x^{b_1}(1),\ldots,x^{b_1}(i) $, then $ p^{b_p}(i)=j $ (or more precisely, the binary representation of $ j $).
	\item Counter for number of nonzero blocks: $ c^{\log \beta} $ is set to the number of nonzero blocks.
	\item Memory table: For every $ i $, if $ s_i=1 $ and $ p^{b_p}(i)=j $, then $ y^{b_1}(j,1)=x^{b_1}(i) $ and $ y^{\log(b/b_1)}(i,2) $ is equal to the binary representation of $ i $.
\end{itemize}

\subsubsection{Local decoding}
Suppose that we want to recover $ x_{i} $ which happens to be the $ i_1 $th bit in the $ i_2 $th block (\emph{i.e.,} $i=(i_2-1)b_1+i_1$).
\begin{itemize}
	\item If $ s_{i_2}=0 $, then output $ 0 $. This is because $ s_{i_2}=0 $ implies that the entire block is zero.
	\item If not, then read $ p^{b_p}(i_2) $. If $ p^{b_p}(i_2)=j $, then output the $ i_1 $th bit in $ y^{b_1}(j,1) $.
\end{itemize}
The maximum number of bits probed is 
\[
\rwc = 1+b_p+1 = 2+\log\beta.
\]
\subsubsection{Local update}
Suppose that we want to update $ x_{i} $ (which happens to be the $ i_1 $th bit in the $ i_2 $th block) with $ \widetilde{x}_i $.
The update algorithm works as follows:
\begin{itemize}
	\item Suppose that $ x_i=0 $ and $ \widetilde{x}_i =1$. The updater first reads $ s_{i_2} $. 
	\begin{itemize}
		\item If $ s_{i_2}=1 $, then it reads $ p^{b_p}(i_2) $. Suppose that $ p^{b_p}(i_2) =j$. Then it writes  $ \widetilde{x}_i $ into the $ i_1 $th location of $ y^{b_1}(j,1) $.
		\item If $ s_{i_2}=0 $, then it means that the block was originally a zero block. The updater sets $ s_{i_2} $ to $ 1 $, and increments the counter for the number of nonzero blocks $ c^{\log\beta} $ by~$ 1 $. Suppose that after incrementing, $ c^{\log\beta}=j $. Then the updater sets $ p^{b_p}(i_2)=j $, writes $ \widetilde{x}^{b_1}(i_2) $ into $ y^{b_1}(j,1) $, and sets $ y^{\log(b/b_1)}(j,1) $ to $ i_2 $.
	\end{itemize} 
	\item Suppose that $ x_i=1 $ and $ \widetilde{x}_i =0$. The updater first reads $ s_{i_2} $. Clearly, this should be equal to $ 1 $. The updater reads $ p^{b_p}(i_2) $ (suppose that it is equal to $ j $), and then $ y^{b_1}(j,1) $ to compute $ x^{b_1}(i_2) $. 
	\begin{itemize}
		\item If $ x^{b_1}(i_2) $ has Hamming weight greater than $ 1 $, then it flips the $ i_1 $th bit of $ y^{b_1}(j,1) $.
		\item If not, then it implies that $ \widetilde{x}^{b_1}(i_2)=0^{b_1} $. The updater next sets $ s_{i_2} $ to $ 0 $. It then decrements $ c^{\log \beta} $. It next overwrites $ y^{b_m}(j) $ with the contents of $ y^{b_m}(c^{\log \beta}) $, and sets $ p^{b_p}(y^{\log(b/b_1)}(c^{\log \beta},2)) $ to $ j $. This is to consistently ensure that the first $ c^{\log \beta} $ chunks of the memory table always contains all the information about nonzero blocks.
	\end{itemize}
	The maximum number of bits that need to be read and written in order to update a single message bit is 
	\[
	\twc = 2+ b_p+2\log \beta + 2b_m+b_p = 2+b_1+4\log \beta +\log\frac{b}{b_1}.
	\]
\end{itemize}

\subsubsection{Proof of Lemma~\ref{lemma:succ_ds_ologn}}
Let us now prove the statement. We use the above scheme with $ b_1=O\left(\frac{1}{\delta}\log b\right) $.
From~\eqref{eq:space_newds}, the total space used is 
\[
k_{ds}\leq \delta b(1+o(1)).
\]
The worst-case local decoding is equal to $O(\log b)$
and the worst-case update efficiency is equal to $O\left( \frac{1}{\delta}\log b \right)$.
This completes the proof. \qed

\begin{remark}
The succinct data structure here satisfies (A$1$) but not (A$2$). Clearly, the order of the chunks in the memory table depends on the sequence of updates performed previously. For example, the first chunk could initially contain data of the first subblock. After a number of updates (e.g., involving setting all bits of the first block to zero, inserting bits in other blocks, and then repopulating the first block with ones), the first block could be stored in chunk $ l>1 $.
\end{remark}


\section{Lower bounds on simultaneous locality}\label{sec:lbound_simult_loc}

To obtain lower bounds, we introduce an additional parameter that we might be interested in minimizing: 
the worst-case \emph{local encodability}, $ \ewc $,  defined to be the maximum number of input symbols that any single codeword bit can depend on. Note that this is different from the update efficiency. This was studied in~\cite{mazumdar2017semisupervised}, where the authors related this quantity to a problem of semisupervised learning. It has been established in the literature that separately, each of $ \rwc,
\twc,\ewc $ can be made $ O(1) $ for near-entropy compression. However, it is not known if $$ (\rwc,\twc,\ewc)=(\Theta(1),\Theta(1),\Theta(1)) $$ can be simultaneously achieved.

In this section, we assume (A1)--(A3) and that the scheme is nonadaptive.



\begin{figure}
	\begin{center}
		\includegraphics[width=4.5cm]{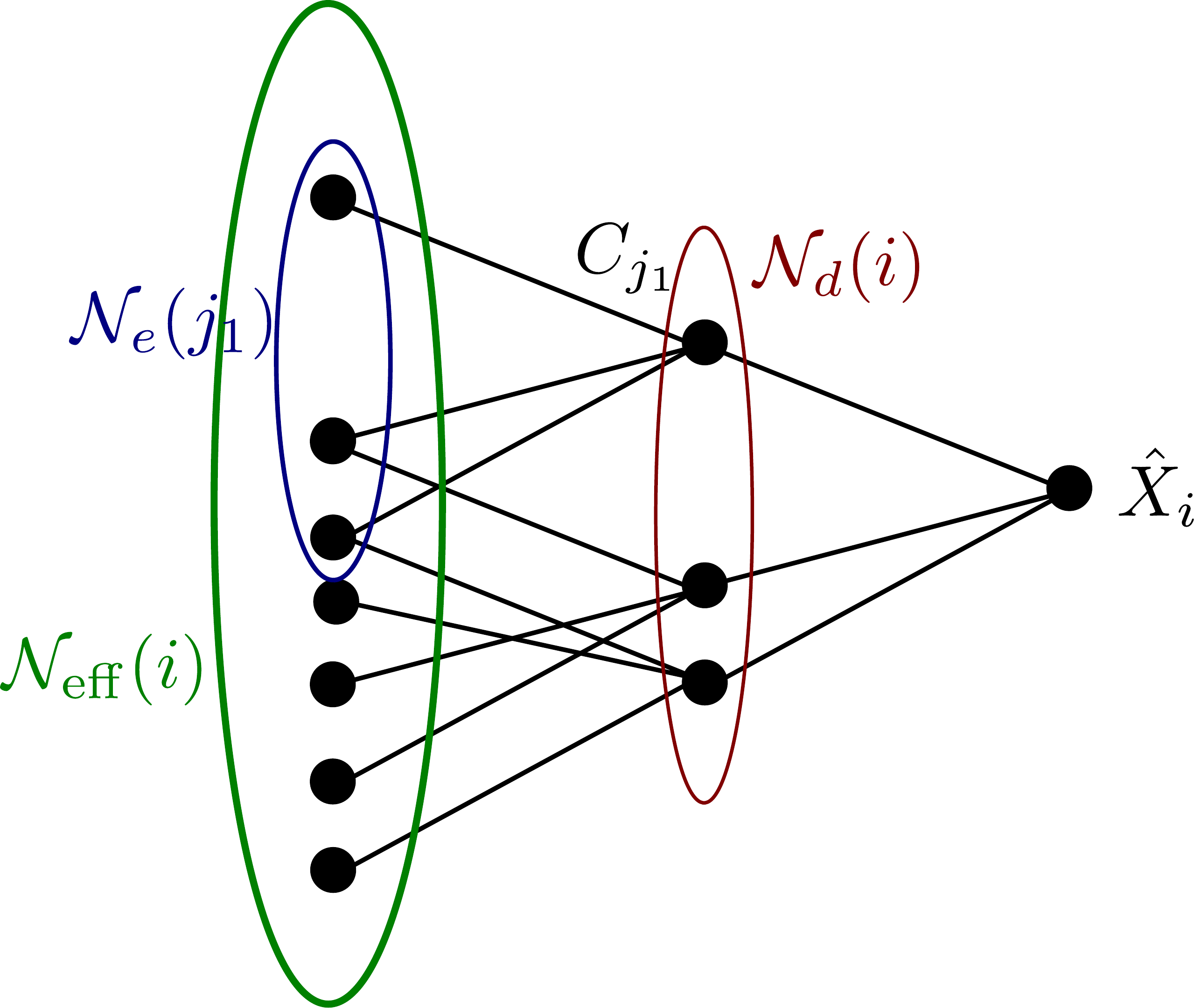}
		\caption{Illustrating various neighbourhoods used in the proofs.}
		\label{fig:graph_neighbourhoods}
	\end{center}
	\vspace{-0.5cm}
\end{figure}


\subsubsection{The local encoding and decoding graphs} 
Our derivation of the lower bound involves analyzing the connectivity properties of two bipartite graphs that describe the local encoding and decoding functions. As we will demonstrate, the various probabilities of error are influenced by the degrees of these bipartite graphs, whose values are governed by $\rwc$, $\ewc$ and $\twc$.
\begin{itemize}
	\item Under (A1) and (A2), the $ j $th compressed bit $ C_j $ can be written as $ C_j = f_j(X_{\cN_e(j)}) $ for some function $ f_j $, where $ \cN_e(j) $ is the set of message locations that $ C_j $ can depend on. 
	\item The $ i $th decoded bit $ \hat{X}_i $ can be written as $ \hat{X}_i =g_i(C_{\cN_d(i)})$ for some function $ g_i $, where $ \cN_d(i) $ is the set of codeword locations that need to be probed in order to recover $ X_i $.
	\item We can construct an $ n\times nR $ encoder bipartite graph $ \cG_e $ where $ (i,j) $ is an edge only if $ i\in \cN_e(j) $, and an $ nR\times n $ decoder bipartite graph $\cG_d$ where $ (j,i)\in [nR]\times n $ is an edge if $ j\in \cN_d(i) $. See Fig.~\ref{fig:encoding_decoding_graph} for an illustration.
	
	\item Let $e_j\defeq |\cN_e(j)|$ be the degree of the $j$th right vertex in $\cG_e$.
	This is equal to the local encodability of the $i$th codeword symbol and satisfies $ e_i\leq \rwc $.
	\item  This gives a natural lower bound on the update efficiency for the $ i $th message symbol: it must be greater than or equal to $\Delta_i^l$, the degree of the $i$th left vertex in $ \cG_e $. 
	
	The average degree of a  vertex in the left (corresponding to a message symbol) is equal to $ R $ times the average degree of a right vertex (which corresponds to codeword bits). This implies that $ \twc $ is lower bounded by the average (arithmetic mean of) the individual local encodabilities of the individual codeword bits.  
	\item Likewise, we denote the degree of the $j$th left vertex of $\cG_d$ by $\Delta_j^r$ and that of the $i$th right vertex in $\cG_d$ by $d_i$. Note that $d_i$ is equal to the local decodability of $\widehat{X}_i$ (maximum number of codeword bits to be probed to recover $\widehat{X}_i$). By definition, $d_i\leq \rwc$ for all $i$.
	\item We also define the so-called effective neighbourhood of $\widehat{X}_i$ to be $$ \cN_{\mathrm{eff}}(i) \defeq \{ X_l:l\in\bigcup_{j\in \cN_d(i)} \cN_e(j) \} .$$ The $ i $th decoded symbol is a function of only those symbols in $ \cN_{\mathrm{eff}}(i) $, i.e., there exists a function $ h_i $ such that  $ \hat{X}_i = h_i(X_{\cN_{\mathrm{eff}}(i)})$. This is illustrated in Fig.~\ref{fig:graph_neighbourhoods}.
\end{itemize}

Let us now obtain a lower bound on the bit error probability of any compression scheme satisfying the above properties. We would like to point out that the following Lemmas~\ref{lemma:lbound_biterror} and~\ref{lemma:local_encodedecdesimult_lb} only make use of assumptions (A1) and (A2), and do not require (A3). 
\begin{lemma}\label{lemma:lbound_biterror}
	The probability of bit error, $ P_e^{(i)}\defeq \Pr[\hat{X}_i\neq X_i] $ satisfies
	\[
	P_e^{(i)} =\begin{cases}
	\geq (1-p)^{|\cN_{\mathrm{eff}}(i)|}\ge (1-p)^{\ewc\rwc}, & \text{ or,}\\
	0. 
	\end{cases}
	\]
\end{lemma}
\begin{proof}
	For every $ i\in[n] $, the decoded symbol $ \hat{X}_i $ is a deterministic function (the composition of $ g_i $ and $ f_j $'s) of $ \cN_{\mathrm{eff}}(i) \defeq \{ X_l:l\in\bigcup_{j\in \cN_d(i)} \cN_e(j) \} $. But we have $ |\cN_{\mathrm{eff}}(i)|\leq \ewc\rwc $.
	The probability of error is given by\footnote{For an event $ \cE $, $ 1_{\cE} $ is the indicator function which takes value $ 1 $ if $ \cE $ occurs, and zero otherwise.}
	\begin{align*}
	P_e^{i} &= \sum_{x}\left(\prod_{l\in \cN_{\mathrm{eff}}(i)}\Pr[X_l = x_l]\right) 1_{\{X_i=\hat{X}_i\}} \\
	&\geq \sum_{x}\left(\prod_{l\in  \cN_{\mathrm{eff}}(i)}(1-p)\right) 1_{\{X_i=\hat{X}_i\}}.
	\end{align*}
	If there is even a single configuration of $ X^n $ for which $ X_i\neq \hat{X}_i $, then $ P_e^{i}\geq (1-p)^{|\cN_{\mathrm{eff}}(i)|}\geq (1-p)^{\ewc\rwc} $.
\end{proof}

\subsubsection*{Remark}
Observe that Lemma~\ref{lemma:lbound_biterror} holds even without assumption (A$3$). Also, if $R<1$, then by the pigeonhole principle $P_e^{(i)}>0$ for a nonvanishing fraction of $i$'s. Since $\peglob\geq \max_iP_e^{(i)}$, we can conclude that any compression scheme satisfying (A$1$)--(A$2$) and having vanishing probability of error and nontrivial rate must satisfy $\twc\twc=\omega(1)$. In other words, one cannot achieve $(\rwc,\twc)=(O(1),O(1))$.

A similar observation was made in~\cite{makhdoumi_onlocallydecsource} for linear locally decodable (but not updatable) compression schemes. If the encoder of the compressor must be linear, then $\rwc=\omega(1)$ if we desire a vanishingly small probability of error.


We can use the above lemma to prove a lower bound on simultaneous local encodability and decoding.
\begin{lemma}\label{lemma:local_encodedecdesimult_lb}
	Any fixed-length compression scheme achieving $ R<1, $ a vanishing probability of error, and satisfying assumptions (A1) and (A2) must necessarily satisfy
	\[
	\ewc\rwc  = \Omega(\log n)
	\]
\end{lemma}
\begin{proof}
	If we operate at a rate $ R<1 $, then by the pigeonhole principle, it must be the case that at least $ n(1-R) $ bits have a nonzero bit error probability. Call this set of bits/locations $ \cE $.
	
	We will now select a subset $ \cS $ of $ O(n/\ewc\rwc) $ many $ X_i $'s whose probabilities of error are all nonzero but whose error events are statistically independent of each other. We can construct such an $ \cS $ in a greedy manner.
	
	For every $ i\in\cE $, the corresponding $ \hat{X}_i $ is a deterministic function of $ \{X_l:l\in \cN_{\mathrm{eff}}(i)\} $. In other words, each $ \hat{X}_i $ depends on at most $ \ewc\rwc $ many $ X_l $'s. We can start with the empty set, and iteratively put those $ X_i $'s in $ \cS $ which are in $ \cE $ but not in $ \bigcup_{i_1\in \cS} \cN_{\mathrm{eff}}(i_1)$. It is clear that $$ |\cS|\geq n(1-R)/(\rwc\ewc) .$$ Moreover, by construction, we have that all the events $ 1_{X_{i}\neq \hat{X}_i} $ are independent for all $ i \in\cS$ (since the corresponding sets of $ X_l $'s they depend on are disjoint).
	
	Let $ m\defeq n(1-R)/(\ewc\rwc) $.
	Using Lemma~\ref{lemma:lbound_biterror}, the probability of block error is bounded as
	\begin{align*}
	P_e &\geq 1-(1-(1-p)^{\rwc\ewc})^{m} \geq 1-e^{-m(1-p)^{\rwc\ewc}}\\
	&\geq m(1-p)^{\rwc\ewc}.
	\end{align*}
	The above quantity is vanishing in $ n $ only if $ \rwc\ewc = \Omega(\log n) $
\end{proof}

\subsection{Lower bound on $ \rwc\twc $ and proof of Theorem~\ref{thm:dwc_uwc_lbound_adaptive}}
We can in fact tighten the analysis above to prove an identical lower bound on $ \rwc\twc $.


Let $ \eta_i\coloneq \cN_{\mathrm{eff}}(i) $. Let $ u_i,e_j,d_i $ respectively denote the update efficiency for the $ i $th message symbol, the local encodability of the $ j $th codeword symbol, and the local decoding of the $ i $th message symbol. These are also respectively greater than or equal to the degrees of the $ i $th left vertex in $ \cG_e $, the $ j $th right vertex in $ \cG_e $, and the $ i $th right vertex in $ \cG_d $.

\begin{lemma}\label{lemma:eta_i_bound}
	Consider any fixed-length compression scheme achieving vanishing probability of error and nontrivial compression rate $ R<1 $.
	There exists a set $ \cS\subset[n] $ of message symbols with  $ \cS = \Theta(n) $ such that for all $ i\in \cS $,
	\[
	 \eta_i =\Omega(\log n).
	\]  
\end{lemma} 
The proof is very similar to that of Lemma~\ref{lemma:local_encodedecdesimult_lb}, so we only sketch the details. We select a subset $ \cS' $ of $ [n] $ such that $ P_e^i>0 $ for all $ i\in \cS' $. We know that $ |\cS'|\geq n(1-R) $. 
A more careful rederivation of  Lemma~\ref{lemma:lbound_biterror} gives us $ P_e^i \geq (1-p)^{\eta_i} $ for all $ i\in\cS' $. Hence,
\begin{align}
P_e &\geq 1- \prod_{i=1}^n(1-P_e^i) \\
     & \geq \sum_{i\in\cS'} (1-p)^{\eta_i},
\end{align}
which is nonvanishing in $ n $ if $ \eta_i <\log n $ for any subset of $ \cS' $ of size $ n(1-R)/2 $. Therefore, there must  exist a subset $ \cS\subset \cS' $ of size $ n(1-R)/2 $ where $ \eta_i >\log n$ for all $ i\in \cS $. This completes the proof. \qed

This leads us to the following result:
\begin{theorem}\label{thm:dwc_uwc_lbound}
	Any fixed-length compression scheme achieving vanishing probability of error and  $ R<1 $, and satisfying assumptions (A1)---(A3) and nonadaptive local algorithms must have
	\[
	\rwc\twc = \Omega(\log n).
	\]
\end{theorem}
\begin{proof}
	From Lemma~\ref{lemma:eta_i_bound}, there exists a set $ \cS $ of size $ \Theta(n) $ such that for all $ i\in\cS $, $ \eta_i  = \Omega(\log n) $. However, $ \eta_i\leq \rwc \sum_{j\in\cN_d(i)}e_j $. From the graph $ \cG_e, $ we also have 
	\begin{equation}\label{eq:dwcuwc_1}
	\rwc\sum_{j\in [nR]}e_j = \rwc \sum_{i\in [n]}\Delta_i^l \leq   \rwc \sum_{i\in [n]}u_i \leq n\rwc\twc, 
	\end{equation}
	where $ \Delta_i^l $ denotes the degree of the $ i $th left vertex in $ \cG_e $.
	
	Using Lemma~\ref{lemma:eta_i_bound}, we have 
	\[
	 |\cS| \Omega(\log n) \leq \sum_{i\in \cS} \eta_i \leq \sum_{i=1}^n \eta_i \leq \rwc \sum_{j\in[nR]} e_j.
	\]
	From \eqref{eq:dwcuwc_1} and using the fact that $ |\cS| = \Theta(n) $, we have
	$
	n\rwc\twc = \Omega(n\log n)  
	$,
	which implies 
	$
	\rwc\twc = \Omega(\log n),
	$
	completing the proof.
\end{proof}

 Using Theorem~\ref{thm:dwc_uwc_lbound} and our remark on adaptive vs nonadaptive schemes just prior to Theorem~\ref{thm:dwc_uwc_lbound_adaptive}, we get Theorem~\ref{thm:dwc_uwc_lbound_adaptive}.

\bibliographystyle{ieeetr}
\bibliography{locality_references}
\end{document}